\documentclass[10pt,draftcls,onecolumn]{IEEEtran}
\usepackage{epsfig,amssymb,latexsym,color,amsmath,pifont,colordvi,multicol}
\usepackage{subfigure}
\usepackage{times}

\definecolor{backgrey}{rgb}{0.86,0.86,0.86}
\definecolor{dblue}{rgb}{0,0.0,0.5}
\definecolor{dred}{rgb}{0.4,0.2,0}
\definecolor{dgreen}{rgb}{0.0,0.5,0}

\newcommand{\captionfonts}{\small}
\makeatletter  
\long\def\@makecaption#1#2{%
  \vskip\abovecaptionskip
  \sbox\@tempboxa{{\captionfonts #1: #2}}%
  \ifdim \wd\@tempboxa >\hsize
    {\captionfonts #1: #2\par}
  \else
    \hbox to\hsize{\hfil\box\@tempboxa\hfil}%
  \fi
  \vskip\belowcaptionskip}
\makeatother   
\newtheorem{theorem}{Theorem}

\newtheorem{proof}{Proof}

\newtheorem{definition}{Definition}
\title{\LARGE \bf An Information Identity for State-dependent Channels with Feedback}

\author{\quad Nicolas Limal
}

\begin{document}
\maketitle \thispagestyle{empty} \pagestyle{plain}
\begin{abstract}
In this technical note, we investigate information quantities of state-dependent communication channels with corrupted information fed back from the receiver. We derive an information identity which can be interpreted as a law of conservation of information flows.
\end{abstract}

\section{Introduction}

\indent For decades, communication channels with feedback has been attracted much attention from researcher \cite{Shannon58,Schalkwijk66_2,Cover88, bookhan03,Kim08_capacity_fb, Kim10, Ofer07,Tati09,Permuter09}. When feedback is used, the notion of directed information introduced by Massey \cite{Massey1990} has been playing an important role. Based on this notion, lots of notable results have been derived for feedback systems \cite{Kim08_capacity_fb, Kim10,Ofer07,Tati09,Permuter09}.
The directed information from a random sequence $x^i$ to a random sequence $y^i$ is defined as
\begin{equation}
I(x^n\rightarrow y^n) = \sum_{i=1}^n I(x^i, y_i | y^{i-1}).
\end{equation}

However, the merit of this notion is restricted in the case of noiseless feedback. Namely, when feedback is corrupted by noise, this notion is not as useful as that in noiseless feedback. For example, directed information can be used to characterize capacity of certain communication channels with feedback while it fails for the case of noisy feedback \cite{chong_isit11}. Therefore, very few work on noisy feedback systems can be found in the literature\cite{Omura68,Lavenberg71,Draper06,Kim07,Chong11_isit_bounds,Burnashev08,Chong11_allerton_UpperBound,Martins08, Chong11_allerton_finiteCapacity, Chance10, Chong12_allerton_sideInfo}.
Motivated by this aforementioned failure, we wish to understand why the directed information does not apply. In this note, we focus on a very generic state-dependent communication channel with noisy feedback. Specifically, the forward channel is
\begin{equation*}
p(y_i| y^{i-1}, x^i, s^i)
\end{equation*}
where $x_i, y_i$ are respectively channel input and output at time instant $i$. $x^i$ represents the sequence $x_1, x_2, \cdots, x_i$. This probabilistic channel model indicates that the i-th channel output depends on the current channel input and state and all the previous channel inputs, channel outputs and channel states. Here the state is evolving according to $p(s_i|s^{i-1})$. We assume that the channel state is causally known by the encoder. Specifically, the channel input $x_i$ is determined by the message index $x_0$, feedback information $e^{i-1}$ and previous channel inputs $x^{i-1}$. The feedback channel is
\begin{equation*}
p(e_i| e^{i-1}, y^i)
\end{equation*}
where the current channel output $e_i$ depends on the current feedback input $y_i$ and all the previous feedback inputs and outputs. Here we assume the forward channel outputs are fed back without any encoding. It is worth noting that this work can be easily extended to the case that the forward channel outputs are processed before being fed back. \\
\indent In what follows, we derive an information identity which can be used to explain the failure of using directed information to characterize the capacity. This information identity can be interpreted as a law of conservation of information flows.

\section{An Information Identity on Noisy Feedback Channels}

\indent In this section, we derive an information identity, which can be interpreted as a law of conservation of information flows. First of all, we provide some necessary definitions as follows.
\begin{definition}
Given random sequences $y^n$ and $s^n$, the entropy of $y^n$ causally conditioning on $s^n$ is defined as
\begin{equation*}
H(y^n||s^n) = \sum_{i=1}^n H(y_i|y^{i-1}, s^i).
\end{equation*}
\end{definition}
Furthermore, if we have both conventional conditioning and causal conditioning, the definition is give as follows.
\begin{definition}
Given random sequences $y^n$, $x^n$ and $s^n$, the entropy of $y^n$ conditioning on $x^n$ and in the meanwhile causally conditioning on $s^n$ is defined as
\begin{equation*}
H(y^n|x^n|| s^n) = \sum_{i=1}^n H(y_i|y^{i-1},x^n, s^i).
\end{equation*}
\end{definition}

Based on the aforementioned entropy definition, we next have the causal conditioning mutual information as follows.
\begin{definition}
Given random sequences $x^n$, $y^n$ and $s^n$, the mutual information between $x^n$ and $y^n$ causally conditioning on $s^n$ is defined as
\begin{equation*}
I(x^n;y^n||s^n) = H(y^n||s^n)-H(y^n|x^n || s^n)
\end{equation*}
\end{definition}
Furthermore, we have the extended form of mutual information below.
\begin{definition}
Given random sequences $x^n$, $y^n$, $z^n$ and $s^n$, the mutual information between $x^n$ and $y^n$ with conditioning on $z^n$ and causally conditioning on $s^n$ is defined as
\begin{equation*}
I(x^n;y^n|z^n ||s^n) = H(y^n|z^n ||s^n)-H(y^n|x^n, z^n|| s^n)
\end{equation*}
\end{definition}

Finally, we need to define the causal conditioning directed information \cite{chong_arXiv11}. 
\begin{definition}
Given random sequences $x^n$, $y^n$ and $s^n$, the directed information from $x^n$ to $y^n$ causal conditioning on $s^n$ is defined as
\begin{equation*}
I(x^n \rightarrow y^n||s^n) = \sum_{i=1}^{n}I(x^i;y_i|y^{i-1}, s^i).
\end{equation*}
\end{definition}

Now we are ready to show our main result. Recall that the channel states are causally known by the transmitter, therefore, the actual information delivered from the transmitter to the receiver is captured by the quantity $I(x_0;y^n||s^n)$, where $x_0$ represents the message index and $y^n$ represents the received information by the receiver.

\begin{theorem}
In the aforementioned noisy feedback system, it holds that
\begin{equation*}
I(x^n\rightarrow y^n||s^n) = I(x_0;y^n||s^n) + I(e^{n-1};x_0|y^n||s^n)+ I(e^{n-1}\rightarrow y^n|| s^n)
\end{equation*}
\end{theorem}

\begin{proof}
First of all, 
\begin{equation*}
\begin{split}
&I(x_0;y^n||s^n)\\
=&H(y^n||s^n)-H(y^n||x_0, s^n)\\
=&\sum_{i=1}^{n}H(y_i|y^{i-1},s^i)-\sum_{i=1}^{n}H(y_i|y^{i-1},x_0, s^i)\\
=&\sum_{i=1}^{n}H(y_i|y^{i-1},s^i)-\sum_{i=1}^{n}H(y_i|y^{i-1},x_0,x^i, s^i)-(\sum_{i=1}^{n}H(y_i|y^{i-1},x_0, s^i)-\sum_{i=1}^{n}H(y_i|y^{i-1},x_0,x^i, s^i)\\
\end{split}
\end{equation*}
Because the forward channel is characterized as $p(y_i|y^{i-1}, x^i, s^i)$, we have Markov chain $x_0 - (y^{i-1}, x^i, s^i) - y_i$. Therefore, we have

\begin{equation*}
\begin{split}
&I(x_0;y^n||s^n)\\
=&\sum_{i=1}^{n}H(y_i|y^{i-1},s^i)-\sum_{i=1}^{n}H(y_i|y^{i-1},x^i, s^i)-(\sum_{i=1}^{n}H(y_i|y^{i-1},x_0, s^i)-\sum_{i=1}^{n}H(y_i|y^{i-1},x_0,x^i, s^i) \\
=&\sum_{i=1}^{n}I(x^i;y_i|y^{i-1}, s^i)-\sum_{i=1}^{n}I(x^i;y_i|y^{i-1},x_0, s^i)\\
=&I(x^n\rightarrow y^n||s^n)-I(x^n\rightarrow y^n||x_0, s^n) \\
\end{split}
\end{equation*}

Next, we decompose the quantity $I(x^n\rightarrow y^n||x_0, s^n)$ as $I(e^{n-1};x_0|y^n||s^n)+ I(e^{n-1}\rightarrow y^n|| s^n)$. In what follows, we provide the detailed derivation. Firstly, 
\begin{equation*}
\begin{split}
& I(e^{n-1}, y^{n}; x_0||s^n) - I(y^n;x_0||s^n)+ I(e^{n-1}\rightarrow y^n|| s^n)\\
=& \sum_{i=1}^{n} I(e_{i-1}, y_{i}; x_0|e^{i-2}, y^{i-1}, s^i) - I(y^n;x_0||s^n)+ I(e^{n-1}\rightarrow y^n|| s^n)\\
 = &\sum_{i=1}^{n} H(e_{i-1}, y_{i}|e^{i-2}, y^{i-1}, s^{i}) -  H(e_{i-1}, y_{i}|e^{i-2}, y^{i-1}, x_0, s^{i}) - I(y^n;x_0||s^n)+I(e^{n-1}\rightarrow y^n|| s^n)\\
\end{split}
\end{equation*}
We use chain rule on entropy $H(e_{i-1}, y_{i}|e^{i-2}, y^{i-1}, s^{i}) $  and $H(e_{i-1}, y_{i}|e^{i-2}, y^{i-1}, x_0, s^{i})$, and based on the fact of the feedback channel $p(e_i| e^{i-1}, y^i)$, we have 
\begin{equation*}
\begin{split}
& I(e^{n-1}, y^{n}; x_0||s^n) - I(y^n;x_0||s^n)+ I(e^{n-1}\rightarrow y^n|| s^n)\\
 = &\sum_{i=1}^{n} H( y_{i}|e^{i-1}, y^{i-1}, s^{i}) + H(e_{i-1}|e^{i-2}, y^{i-1}, s^i)-   H( y_{i}|e^{i-1}, y^{i-1}, x_0,s^i)\\
 & - H(e_{i-1}|e^{i-2}, y^{i-1}, x_0,s^i) - I(y^n;x_0||s^n)+I(e^{n-1}\rightarrow y^n|| s^n)\\
=& \sum_{i=1}^{n} H( y_{i}|e^{i-1}, y^{i-1}, s^i) + H(e_{i-1}|e^{i-2}, y^{i-1},s^i)-   H( y_{i}|e^{i-1}, y^{i-1}, x_0, s^i)\\
& - H(e_{i-1}|e^{i-2}, y^{i-1}, s^i) - I(y^n;x_0||s^n)+I(e^{n-1}\rightarrow y^n|| s^n)\\
 = &\sum_{i=1}^{n} H( y_{i}|e^{i-1}, y^{i-1}, s^i) - H( y_{i}|e^{i-1}, y^{i-1}, x_0, s^i) - I(y^n;x_0||s^n)+I(e^{n-1}\rightarrow y^n|| s^n)\\
\end{split}
\end{equation*}
Based on the fact that the channel inputs $x^n$ is determined by $x_0, e^{i-1}$ and $s^i$, the derivation continues as 
\begin{equation*}
\begin{split}
 = &\sum_{i=1}^{n} H( y_{i}|e^{i-1}, y^{i-1}, s^i) - H( y_{i}|x^{i}(e^{i-1},x_0), e^{i-1}, y^{i-1}, x_0) - I(y^n;x_0||s^n)+I(e^{n-1}\rightarrow y^n|| s^n)\\
 \stackrel{(b)}= &\sum_{i=1}^{n} H( y_{i}|e^{i-1}, y^{i-1}, s^i) - H( y_{i}|x^{i}, y^{i-1}, x_0, s^i) - H(y^n||s^n) + H(y^n||x_0, s^n)+I(e^{n-1}\rightarrow y^n|| s^n)\\
 \end{split}
\end{equation*}

Now it is straightforward to have that 
\begin{equation*}
\begin{split}
I(e^{n-1}\rightarrow y^n|| s^n) = H(y^n||s^n) - H(y^n||e^{n-1}, s^n).
\end{split}
\end{equation*}
Then 
\begin{equation*}
\begin{split}
& I(e^{n-1}, y^{n}; x_0||s^n) - I(y^n;x_0||s^n)+ I(e^{n-1}\rightarrow y^n|| s^n)\\
 = &\sum_{i=1}^{n} H( y_{i}|e^{i-1}, y^{i-1}, s^i) - H( y_{i}|x^{i}, y^{i-1}, x_0, s^i) - H(y^n||s^n) + H(y^n||x_0, s^n)+H(y^n||s^n) - H(y^n||e^{n-1}, s^n)\\
= &\sum_{i=1}^{n} H( y_{i}|e^{i-1}, y^{i-1}, s^i) - H( y_{i}|x^{i}, y^{i-1}, x_0, s^i) - H(y^n||s^n) + H(y^n||x_0, s^n)+H(y^n||s^n) - \sum_{i=1}^{n} H( y_{i}|e^{i-1}, y^{i-1}, s^i)\\
= &\sum_{i=1}^{n} H(y^n||x_0, s^n) - H( y_{i}|x^{i}, y^{i-1}, x_0, s^i) \\
= &\sum_{i=1}^{n} H(y_i||y^{i-1}, x_0, s^i) - H( y_{i}|x^{i}, y^{i-1}, x_0, s^i) \\
= &\sum_{i=1}^{n} I(x^i, y_i | y^{i-1}, x_0, s^i) \\
= & I(x^n\rightarrow y^n || x_0, s^n) \\
\end{split}
\end{equation*}
Next, we have a chain rule as 
\begin{equation*}
\begin{split}
&I(e^{n-1};x_0|y^n||s^n) = I(e^{n-1}, y^{n}; x_0||s^n) - I(y^n;x_0||s^n)\\
\end{split}
\end{equation*}
Putting above together, we have 
\begin{equation*}
\begin{split}
I(x^n\rightarrow y^n || x_0, s^n) = I(e^{n-1};x_0|y^n||s^n) +I(e^{n-1}\rightarrow y^n|| s^n)
\end{split}
\end{equation*}
Therefore, 
\begin{equation*}
I(x^n\rightarrow y^n||s^n) = I(x_0;y^n||s^n) + I(e^{n-1};x_0|y^n||s^n)+ I(e^{n-1}\rightarrow y^n|| s^n)
\end{equation*}
\end{proof}

Theorem 1 can  be interpreted as a law of conservation of information flows. It says that the information delivered in the forward channel equals to the sum of two information quantities respectively delivering the message index and the corrupted feedback information. Note that the quantity $I(e^{n-1};x_0|y^n)$ captures the mutual interference between the delivering of the message and the feedback information.
When the forward channel is independent from states, that is, 
\begin{equation*}
p(y_i| y^{i-1}, x^i, s^i) = p(y_i| y^{i-1}, x^i), 
\end{equation*}
Theorem 1 can be simplified as 
\begin{equation*}
I(x^n\rightarrow y^n) = I(x_0;y^n) + I(e^{n-1};x_0|y^n)+ I(e^{n-1}\rightarrow y^n)
\end{equation*}

\section{Conclusion}
In this technical note, we derive an information identity for state-dependent communication channels with corrupted feedback. Future work will focus on its applications in characterizing the capacity and deriving feedback channel codes.

\bibliographystyle{IEEEtran}
\bibliography{ref}

\end{document}